\documentclass[runningheads]{llncs}
\usepackage{booktabs} 

\title{\bf Payment Schemes from Limited Information with Applications in Distributed Computing}
\author{Nikolaj I. Schwartzbach}
\institute{Dept. of Computer Science, Aarhus University}

\titlerunning{Payment Schemes from Limited Information}

\usepackage[utf8]{inputenc}
\usepackage{amsmath,amssymb}

\usepackage{booktabs}

\usepackage{hyperref}
\usepackage[capitalise]{cleveref}
\usepackage{esvect}
\usepackage{framed}

\usepackage{tikz}
\usetikzlibrary{arrows.meta}
\usetikzlibrary{calc,shapes.callouts,shapes.arrows}

\newcommand{\U}{\mathbf{U}}
\newcommand{\E}{\mathbf{E}}
\newcommand{\ps}{\textsf{PaymentScheme}_t^\delta}
\newcommand{\lp}{\textsf{LP}}
\newcommand{\K}{\mathbf{K}}
\newcommand{\T}{\mathrm{T}}
\newcommand{\X}{\mathbf{X}}

\renewcommand{\AA}[1]{\mathbf{\Psi}^{(#1)}}
\newcommand{\e}[1]{\boldsymbol{\delta}^{(#1)}}
\newcommand{\gear}[5]{%
	\foreach \i in {1,...,#1} {%
		[rotate=(\i-1)*360/#1]  (0:#2)  arc (0:#4:#2) {[rounded corners=1.5pt]
			-- (#4+#5:#3)  arc (#4+#5:360/#1-#5:#3)} --  (360/#1:#2)
}}  
\newcommand{\norm}[1]{\ensuremath\left \lVert #1 \right \rVert}
\renewcommand{\L}{\mathbf{\Lambda}}
\renewcommand{\P}{\mathbf{\Phi}}
\newcommand{\R}{\mathbb{R}}

\newcommand{\Fpvc}{\mathcal{F}_\texttt{PVC}}

\newcommand{\prob}[3]{\begin{framed}\vspace{-0.2cm}
		\begin{center}
		    #1
		\end{center}\vspace{-0.2cm}\hrulefill\\
		\noindent \textbf{Instance}: \,#2\\
		\noindent \textbf{Output}: #3\vspace{-0.1cm}
\end{framed}}

\begin{document}

\maketitle
\begin{abstract}
    We propose a generic mechanism for incentivizing behavior in an arbitrary finite game using payments. Doing so is trivial if the mechanism is allowed to observe all actions taken in the game, as this allows it to simply punish those agents who deviate from the intended strategy. Instead, we consider an abstraction where the mechanism probabilistically infers information about what happened in the game. We show that payment schemes can be used to implement any set of utilities if and only if the mechanism can essentially infer completely what happened. We show that finding an optimal payment scheme for games of perfect information is \textsf{P}-complete, and conjecture it to be \textsf{PPAD}-hard for games of imperfect information. We prove a lower bound on the size of the payments, showing that the payments must be linear in the intended level of security. We demonstrate the applicability of our model to concrete problems in distributed computing, namely decentralized commerce and secure multiparty computation, for which the payments match the lower bound asymptotically.
\end{abstract}

\clearpage
\section{Introduction}
Game theory is the study of strategic reasoning in rational agents, where rationality means the agents act to maximize their own utility. When all agents choose such a strategy, we call the resulting interaction an equilibrium. Unfortunately, the equilibrium often does not ensure the best outcome for the agents involved. The most famous example is the prisoner's dilemma where two criminals are arrested and interrogated by police in separate rooms: each criminal can either cooperate with their accomplice, or defect and give them up to the police, resulting in a reduced sentence. Here, it is well-known that cooperation is not an equilibrium, as neither criminal can trust the other not to defect, although it would be in their common interest to do so. In game theory, this inefficiency is measured using the \emph{price of anarchy} (PoA), defined as the ratio of the social optimum and the worst possible equilibrium. In seminal work \cite{poa}, Koutsoupias and Paradimitriou consider a simple model of network routing where the PoA is shown to be $\geq 1.5$. This means the lack of coordination between the agents leads to a $33\%$ loss of performance compared to the optimal setting in which the agents coordinate. While this may be problematic on its own, consequences may be more severe if the interaction we are trying to model is of legal matters or otherwise related to security. Here, a lack of coordination may lead to irreparable damage (such as leak of private information) if some agents deviate from the intended strategy.

Indeed, in cryptography, such complications are the cause of a seemingly irreconcilable gap between the worlds of rational cryptography and the classic cryptographic model: in \cite{halpern_teague}, Halpern and Teague famously show there is no deterministic bounded-time interactive protocol for secure function evaluation on private inputs involving rational agents with a certain class of utility functions, namely agents who prefer to learn the output of the function, but prefer as few other agents as possible learn the output. By contrast, there are simple and efficient \emph{actively secure} protocols when a sufficient subset of the agents are guaranteed to be honest, even when the remaining agents are allowed to deviate arbitrarily \cite{mpc_book}. A weaker notion of security called `covert security' was proposed by Aumann and Lindell in \cite{covert}. Here, agents are allowed to deviate but are caught with some constant non-zero probability. This was extended by Asharov and Orlandi in \cite{pvc} to publicly verifiable covert (PVC) security where a certificate is output that can be verified by a third party to determine if cheating has occurred. The underlying assumption of these protocols is that the cost associated with the risk of being caught outweighs the benefit of deviating. Indeed, the problem of misaligned equilibria is usually mitigated in practice by ensuring appropriate punishment for misbehaving, such as fining deviants, banning them from participating again, or subjecting them to other legal repercussions, effectively changing the utilities of the game to ensure being honest is, in fact, an equilibrium. In our example with the prisoner's dilemma, a criminal who defects might face consequences after the other criminal is released from prison, as the adage goes: ``snitches get stitches''. Sometimes, it is less clear how to punish agents, as when the games are models of interaction on the internet where agents can be anonymous. 

In this work, we consider using \emph{payments} as a generic way to incentivize the participating agents to behave in a certain way. We ask the following question.
\begin{quote}
    \em Can we design a generic mechanism that uses payments to incentivize behavior in an arbitrary game involving rational agents?
\end{quote}
Of course, this is trivial if we allow the mechanism to observe all actions taken by the parties, as this allows the mechanism to punish parties who deviate from the inteded strategy. Instead, we consider an abstraction where the mechanism is only allowed to probabilistically infer information about what happened in the game. To the best knowledge of the authors, this question has not been studied in full generality before, in a way that allows for `plug-and-play' with an arbitrary game with some probabilistic release of information.

\subsection{Related Work}
\subsubsection{Mechanism Design.}
The use of payments to incentivize behavior is a well-studied problem in mechanism design, where the payments are often called `scoring rules'. Such a rule assigns a score (payment) to each outcome of an interaction, and can e.g. be used to elicit truthful responses. In this case, we say the scoring rule is \emph{proper}, examples of which include the quadratic scoring rule and the logarithmic scoring rule \cite{Selten1998,properscoringrules}. A mechanism for which an agent maximizes their utility by reporting their beliefs truthfully is said to be \emph{truthful}. The logarithmic scoring rule is used by Prelec to implement a truthful mechanism for voting in the Bayesian truth serum (BTS) model \cite{prelec}. This was extended to `robust BTS' by Witkowski and Parkes \cite{robust_bts} that instead uses the quadratic scoring rule. Scoring rules are also used in peer prediction methods \cite{peerprediction}, Bayesian markets \cite{bayesmarket} and choice matchings \cite{choicematching}. Payments are also used more generally in the generic Vickrey-Clarke-Groves (VCG) mechanism for obtaining a socially optimal outcome \cite{vickrey,clarke,groves}. However, the VCG mechanism is fundamentally limited to games that involve distributing a set of `items' among a set of players. It is not obvious how this would apply to an arbitrary extensive-form game.

\subsubsection{Distributed Computing.}
There are numerous works in the literature that take advantage of payment schemes to incentivize the participants to behave honestly. Such payment schemes are usually implemented by deploying a smart contract on a blockchain; the protocol starts with each party submitting a `deposit' that is repaid only if they are found to act as intended. The work most related to ours is by George and Kamara \cite{alas} who propose a framework for incentivizing honesty using `adversarial level agreements' that specify damages parties must pay if found to act adversarially. We will show later that their model can be recovered as a special case of our model. Recently, Faust, Hazay, Kretzler and Schlosser propose `financially backed covert security' \cite{fbc} to punish parties who are caught deviating in a PVC protocol. Their work is focused on the cryptographic implementation, and as a result, they do not formally analyze the equilibria induced by their mechanism. In \cite{pvc2pc}, Zhu, Ding and Huang propose a protocol for 2PC that incentivizes honesty using a PVC augmented with a deposit scheme. In \cite{cloud_computing}, Dong et al. propose a protocol that uses deposit schemes to incentivize honesty in outsourcing cloud computations. BitHalo \cite{bithalo} implements an escrow using deposits and multisigs that was analyzed in \cite{bithalo_analysis} by Bigi et al. In \cite{dual_deposit}, Asgaonkar and Krishnamachari propose a smart contract for decentralized commerce of digital goods using dual deposits. This was extended by Schwartzbach to non-digital goods in \cite{schwartzbach} in a way that uses deposits optimistically. Deposits have also been used for `truth-telling mechanisms': in \cite{astraea}, Adler et al. propose a system, Astraea, that uses deposits and rewards to incentivize a group of voters to decide to validity of a proposition. Kleros \cite{kleros} uses a similar mechanism to implement a decentralized court system. 

\subsubsection{Economics.}
In the economics literature, the payment schemes that we study are known as `deposit-refund systems' \cite{deposit_refund}. They are often studied in the context of environmental issues for incentivizing compliance with laws and regulations. In \cite{game_theory_bottles}, Grimes-Casey et al. propose a game-theoretic model using such deposit-refund systems to analyze consumer behavior with refillable plastic bottles. Indeed, deposit-refund systems are currently used in many countries for closing the gap between the marginal private cost and the marginal external cost of disposing of e.g. bottles, batteries, tires, and consumer electronics, see e.g. \cite{deposit_refund_practice} for an overview. Such systems can also be used at a higher level of governance: in \cite{deposit_refund2}, McEvoy studies deposit-refund systems as a means of enforcing nations to comply with international environmental agreements.

\subsection{Our Contributions}
We propose a generic mechanism to incentivize behavior in an arbitrary $n$-player finite games with the use of payments. This is a trivial task if the mechanism can observe all actions taken in the game; instead, our model assumes the mechanism can probabilistically observe actions taken by the agents. We show that payments can be used to implement any set of utilities if and only if the mechanism can essentially infer the entire execution of the game (\cref{lemma:Phi_must_be_left_invertible_to_implement_any_E}). We show that our model generalizes similar models in the literature, such as `adversarial level agreements' by George and Kamara \cite{alas} retained as a special case. We sketch how to implement the payments in a distributed setting by letting agents deploy a smart contract on a blockchain. We demonstrate how to use payments for decentralized commerce to solve the `eBay problem', such that neither buyer nor seller has an incentive to cheat.

We investigate the computational complexity of computing an optimal deposit scheme, in the sense that the payments are minimized. For games of perfect information, we observe that the problem is equivalent to linear programming under logspace reductions, thus showing the following.
\begin{theorem}[Informal]\label{thm:1_informal}
    Finding an optimal deposit scheme for a finite game of perfect information, or showing no suitable deposit scheme exists, is \textsf{P}-complete.
\end{theorem}

\noindent For games of imperfect information, it is well-known that even computing an equilibrium is \textsf{PPAD}-complete, so it is unlikely there is an efficient algorithm for finding an optimal deposit scheme in these cases. As a consequence, we conjecture that finding an optimal deposit scheme for finite games of imperfect information is \textsf{PPAD}-hard.

To showcase the applicability of our model, we apply it to the problem of secure multiparty computation. We show that payments can be used, together with what is known as a `publicly verifiable covert secure protocol' \cite{covert,pvc}, to yield a secure protocol for secure function evaluation involving rational agents.  We stress that this does not violate the impossibility result of Halpern and Teague for the simple reason that they explicitly assume the utilities are not quasilinear, hence not allowing payments.
\begin{theorem}[Informal]\label{thm:2_informal}
    Any function $f$ can be computed with $\delta$-strong game-theoretic security with rational agents by augmenting an $\varepsilon$-deterrent PVC protocol with a payment scheme where each party pays $O(1+\delta / \varepsilon)$.
\end{theorem}

\noindent Finally, we prove a lower bound on the size of the largest punishment (equivalently, \emph{deposit}) for all games that are `self-contained'. We show the punishments must be linear in the size of the desired level of security. Note that this matches asymptotically the bound of \cref{thm:2_informal} since $n,s$, and $\varepsilon$ are constant for any fixed PVC protocol.
\begin{theorem}[Informal]\label{thm:3_informal}
    Any self-contained payment scheme that achieves $\delta$-strong game-theoretic security in a game of $n$ players must have a maximum punishment of size $\Omega(1 + \delta \sqrt{n} / s)$, where $s$ is the number of observable outcomes.
\end{theorem}

\noindent The paper is organized as follows. We start in \cref{deposit_intro} by defining our model of payment schemes. We show how to implement a payment scheme using a smart contract, and prove that payments can be used to implement any set of utilities if and only if the mechanism can essentially infer all information about what happened. In \cref{complexity}, we consider the computational complexity of finding payment schemes and prove \cref{thm:1_informal}. Next, in \cref{mpc}, we apply the framework to secure MPC and prove \cref{thm:2_informal}. Finally, in \cref{lower_bound}, we show a lower bound on the size of the maximum deposits and prove \cref{thm:3_informal}.

\setcounter{theorem}{0}
\setcounter{conjecture}{0}

\section{Preliminaries and notation}

In this section we briefly state some preliminaries needed for the purpose of self-containment, as well as to establish notation. The set of all real vectors with $n$ elements is given by $\mathbb{R}^n$, and the set of all $m \times n$ matrices is given by $\mathbb{R}^{m \times n}$. We use a boldface font to refer to vectors and matrices, and reserve capital symbols $\mathbf{A}$ for matrices, and lowercase symbols $\mathbf{u}$ for vectors. The $n \times n$ identity matrix is denoted $\mathbf{I}_n$. We may denote by $\mathbf{0}$, resp. $\mathbf{1}$ as either matrices or vectors containing only 0, resp. 1 and trust it is clear from the context what we mean. To emphasize the size, we may write e.g. $\mathbf{0}^n$ as the vector $[0,0,\ldots, 0]^\top \in \mathbb{R}^n$. For a matrix $\mathbf{A}=(a_{ij}) \in \mathbb{R}^{m\times n}$, we denote by $\text{vec}(\mathbf{A})=[a_{11},a_{12},\ldots,a_{1n},a_{21},\ldots,a_{mn}]^\top \in \mathbb{R}^{mn}$ the vectorized version of $\mathbf{A}$, resulting from `flattening' the matrix to turn it into a vector. Note that this is a linear operation. If $\mathbf{a} \in \mathbb{R}^{n}, \mathbf{b} \in \mathbb{R}^{n'}$ are two vectors, we denote by $\mathbf{a} \Vert \mathbf{b} \in \mathbb{R}^{n+n'}$ the `concatenation' of $\mathbf{a}$ and $\mathbf{b}$.

\subsection{Game Theory}

We mostly assume familiarity with game theory and refer to \cite{Osborne1994} for more details. We give a brief recap to establish notation. An \emph{extensive-form game} consists of a rooted tree $T$, the leaves of which are labeled with a utility for each player. We denote by $L \subseteq T$ the set of leaves in $T$, and suppose some arbitrary but fixed order on its elements, $\ell_1, \ell_2, \ldots \ell_m$. We assume the existence of an $n \times m$ matrix $\U = (u_{ij}) \in \mathbb{R}^{n \times m}$, called the utility matrix of $G$, that for each player $P_i$ specifies how much utility $u_{ij}$ they receive when the game terminates in the leaf $\ell_j \in L$. The remaining nodes $T \setminus L$ are partitioned into $n$ sets, one belonging to each player. The game is played, starting at the root, by recursively letting the player who owns the current node choose a child to descend into. We stop when a leaf $\ell_j$ is reached, after which player $P_i$ is given $u_{ij}$ utility. A mapping $s_i$ that dictates the moves a player $P_i$ makes is called a strategy for that player, and is said to be pure if it is deterministic, and mixed otherwise. A set of strategies $s=(s_1,s_2,\ldots, s_n)$, one for each player, is called a strategy profile and defines a distribution on the set of leaves in the game. We overload notation and let $u_i(s)$ denote the expected utility for player $P_i$ when playing the strategy profile $s$. If $C \subseteq \{1,2,\ldots n\}$ is a set of indices of players, a coalition, we denote by $-C$ its complement so that we may write a strategy profile $s$ as $s=(s_C, s_{-C})$. As solution concept, we will use a refinement of regular Nash equilibria that takes into account deviations by more than a single party, see \cite{kresilience} for details on this model. Formally, a strategy profile $s^*$ is said to be a \emph{$t$-robust (Nash) equilibrium} if for every strategy $s_C$ with $|C|\leq t$, and every $i \in C$, it holds that: 
$$u_i(s^*) \geq u_i(s_C, s_{-C}^*)$$
A subgame of $G$ is a subtree $G' \subseteq G$ such that whenever $u \in G'$ and $v \in G$ is a child of $u$, then $v \in G'$. A strategy profile that is a $t$-robust equilibrum for every subgame of $G$ is said to be a \emph{$t$-robust subgame perfect equilibrium} (SPE). These definitions suffice for so-called games of perfect information, where at each step, a player knows the actions taken by previous players, though, more generally, we may consider partitioning each set of nodes belonging to a player into \emph{information sets}, the elements of which are sets of nodes that the player cannot tell apart. A game of perfect information is a special case where all information sets are singletons. 

\section{Payment Schemes}
\label{deposit_intro}
In this section, we present our model of games with payment schemes and show when they can be used to ensure it is rational to play an intended strategy. We consider games of perfect information as every such game allows for backward induction to determine an SPE in linear time in the size of the game tree. In general, we should not hope to efficiently determine an optimal payment scheme for games of imperfect information, as it is well-known that finding an equilibrium in these games is \textsf{PPAD}-complete as shown by Daskalakis, Goldberg and Papadimitriou \cite{ppad}. 

We consider a set of $n$ parties $P_1, P_2, \ldots, P_n$ playing a fixed finite extensive-form game $G$ of perfect information. The parties are assumed to be risk-neutral such that they aim to maximize their expected utility. We also assume the parties have quasilinear utilities such that we can use payments to change their incentives. We take as input a unique pure strategy profile $s^*$ that we want the parties to play that we call the \emph{honest strategy profile} for lack of a better term. Note that $s^*$ is required to be pure, since it is impossible to determine (without multiple samples) if a player played a mixed strategy. This has the effect that $s^*$ defines, at each branch in the game, a unique `honest move' that the corresponding party must play. Our goal is to construct a procedure $\Gamma$ that takes as input a game $G$ in a black-box way and produces an equivalent game $\Gamma(G)$ that implements a different utility matrix $\E$ such that $s^*$ is an equilibrium.

\paragraph{Information Structures}  In order to construct the procedure $\Gamma$, we need to be able to infer \emph{something} about what happened during the execution of the game, as otherwise we are simply `shifting' the utilities of the game, not changing the structure of its equilibria. We call a mechanism that enables inferring information from a game an \emph{information structure}. We assume playing the game emits a symbol from a fixed finite alphabet $\Sigma$ of possible outcomes that can be observed by an outside observer. This alphabet serves as a proxy for how the parties acted in the execution of the game. We associate with each leaf of the game a distribution on $\Sigma$. When the game terminates, we sample a symbol according to the distribution and output that symbol. 

\begin{definition}
	An \emph{information structure} for $G$ is a pair $\langle \Sigma, \P \rangle$ where $\Sigma$ is a finite alphabet of symbols with some arbitrary but fixed order on its symbols, $\sigma_1, \sigma_2, \ldots \sigma_s$ for $s=|\Sigma|$, and where $\P=(\phi_{kj}) \in \R^{s \times m}$ is a matrix of emissions probabilities such that every column of $\P$ is a pdf on the symbols of $\Sigma$. $\hfill\diamond$
\end{definition}

\noindent Given a finite game with an information structure, a \emph{payment scheme} $\Gamma$ is a mechanism that can be used to change the utilities of the game. At the end of the game, the payment scheme rewards or punishes the parties based on what was emitted by the information structure. We assume \emph{quasilinearity}, that is, the utilities of the game are given in the same unit as some arbitrarily divisible currency which the payment scheme is able to process. A party $P_i$ is indifferent to obtaining an outcome that gives them $u_{ij}$ utility and receiving $u_{ij}$ money. In other words, we make the implicit assumption that `everything has a price' and intentionally exclude games that model interactions with events that are not interchangeable with money. This circumvents the impossibility result of Halpern and Teague who implicitly assume a fixed total order on the set of possible outcomes. By contrast, quasilinearity allows the payment schemes to alter the order by punishing or rewarding parties with money. 

\begin{definition}\label{def:deposit_scheme}
   A \emph{payment scheme} for $\langle G, \mathcal{I} \rangle$ is a matrix $\L = \{\lambda_{ik}\} \in \R^{n \times s}$, where $\mathcal{I} = \langle \Sigma, \P \rangle$ is an information structure for $G$, and $\lambda_{ik}$ is the utility lost by $P_i$ when observing the symbol $\sigma_k \in \Sigma$ $\hfill\diamond$
\end{definition}

\noindent In our definition, $\L$ is a matrix that explicitly defines how much utility $\lambda_{ik}$ party $P_i$ loses when the payment scheme observes the symbol $\sigma_k \in \Sigma$. When the game is played reaching the leaf $\ell_j$, by quasilinearity the expected utility of party $P_i$ is the utility they would have received in a normal execution, minus their expected loss from engaging with the payment scheme:

\begin{equation}
    \label{eq:utilities}
    \mathbb{E}[\text{$P_i$ utility in leaf $\ell_j$}] = u_{ij} - \sum_{k=1}^s \lambda_{ik} \, \phi_{kj} = [\U - \L \P]_{ij}
\end{equation}

\noindent Correspondingly, the game $\Gamma^\L(G)$ is said to \emph{implement the utility matrix $\E$} if $\E = \U - \L \P$. 

For the remainder of this paper, we study when and how $\L$ can be instantiated to ensure $\Gamma^\L(G)$ implements some $\E$ with desirable properties.

\noindent Note that $\L$ is allowed to contain negative entries which means parties are \emph{compensated}, i.e. receive back more utility from the payment scheme than they initially deposited. Of course, this necessitates that some other party loses their deposit. In general, we may want the payment scheme to `break even', in the sense that all column sums of $\L$ are zero. However, this severely restricts the class of utility matrices that can be implemented.

\begin{lemma}\label{lemma:zero_inflation_U-E_column_sum_zero}
    To implement $\E$ with zero inflation, each column of $\U-\E$ must sum to zero.
\end{lemma}
\begin{proof}
    First of all, note that any $\L$ that implements $\E$ must satisfy $\L\P=\U-\E$. Any such $\L$ can be written as $\L = \X_0 + \K$ where $\X_0$ is a fixed solution and $\K$ is any element in the cokernel of $\P$, i.e. $\K\P = \mathbf{0}$. In order for $\L$ to have zero inflation, it must hold that
    $\mathbf{1}^\T(\X_0 + \K) = \mathbf{0}$,
    which implies that $\mathbf{1}^\T \K = -\mathbf{1}^\T \X_0$. Now assume $\L$ has zero inflation, then we can multiply by $\P$ from the right to yield
    $\mathbf{1}^\T \K \P = -\mathbf{1}^\T \X_0 \P$.
    But $\K$ is an element of the cokernel of $\P$, and $\X_0$ is a solution to the equation, so it must hold that
    $\mathbf{0} = -\mathbf{1}^\T(\U-\E)$
    which implies the columns of $\U-\E$  sum to zero.\qed
\end{proof}

\subsection{Trivial Information Structures} One might be tempted to simply choose a utility matrix $\E$ with some properties that we like and solve for $\L$ to yield a protocol with those properties. Unfortunately, as we will show, this is only possible for information structures that are `trivial' in the sense that they leak all essentially all information about what happened.

\begin{lemma}\label{lemma:Phi_must_be_left_invertible_to_implement_any_E}
	Let $\U \in \R^{n\times m}, \P \in \R^{s \times m}$ be fixed matrices, and let $\Sigma$ be a fixed alphabet of size $|\Sigma|=s$. Then there exists a $\L_\E$ for each $\E \in \R^{n\times m}$ such that $\L_\E$ implements $\E$ if and only if $\P$ is left-invertible.
\end{lemma}
\begin{proof}
	We prove each claim separately:
	\begin{itemize}
		\item[$\Leftarrow$]  If $\P$ is left-invertible, then for any fixed $\E$ we can let $\L_\E := (\U-\E)\,\P^{-1}$ where $\P^{-1}$ is a left-inverse of $\P$. It follows that
		$\U-\L\P = \U - (\U-\E)\P^{-1}\P = \U - \U + \E = \E$,
		which means that $\L_\E$ implements $\E$, as we wanted to show.\\
		\item[$\Rightarrow$] Suppose there is such a $\L_\E$ for each $\E$. This means that we can always find $\L_\E$ that solves $\U-\E = \L_\E\P$. Assume for the sake of contradiction that there are fewer symbols than leaves. This means there must be a leaf for which the deposits is a fixed linear combination of the deposits of the other leaves which means there must be an $\E$ that we cannot implement. But this is a contradiction so we assume there are at least as many symbols as leaves. This means we can choose $\E$ such that $\U-\E$ is left-invertible with left-inverse $\mathbf{F}\in\R^{m\times n}$, which means that
		$\mathbf{F}\L_\E \P = \mathbf{I}_m$. But this means that $\mathbf{F}\L_\E$ is the left-inverse of $\P$, a contradiction.\qed
	\end{itemize}
\end{proof}

\noindent In particular, we can only implement any $\E$ we want if there are at least as many symbols as leaves in the game tree, and that these symbols are not duplicates, i.e. the distributions of symbols across the leaves are linearly independent. Since we are considering pdfs which are normalized, linear independence means the pdfs are pairwise distinct. This means we can only `do what we want' if the smart contract is basically able to infer completely what happened in the execution of the game. Of course, this makes the problem trivial as mentioned, as intuitively, we can simply keep a deposit for all players who deviated from the intended strategy.

\subsection{A Special Case: Adversarial Level Agreements} We now show how the model of `adversarial level agreements' (ALAs) by George and Kamara \cite{alas} can be recovered as a special case of our model. An ALA for a game with $n$ players consists of 1) a description of the intended strategy for each player, and 2) a vector of damages $\mathbf{d} \in \mathbb{R}^n$ that specifies how much utility $\mathbf{d}_i$ party $P_i$ should lose when found to deviate from the intended strategy. Their model does not explicitly consider deviations by more than a single party, so we can state this as an information structure with the following alphabet:
$$
	\Sigma = \{\top, \bot_1, \bot_2, \ldots, \bot_n\}
$$
Here, $\top$ means all parties were honest, and $\bot_i$ means $P_i$ deviated. The emission matrix $\P$ depends on the specific application. An ALA then corresponds to a payment scheme of the following form.
$$
	\L = \begin{pmatrix}
		0 & \mathbf{d}_1 & 0 & \cdots & 0\\
		0 & 0 & \mathbf{d}_2 & \cdots & 0\\
		\vdots & \vdots & \vdots & \ddots & \vdots\\
		0 & 0 & 0 & \cdots & \mathbf{d}_n
	\end{pmatrix}
$$
Note that we can easily generalize this to deviations by any $t\leq n$ parties by including more symbols, e.g. $\bot_{12}$ or $\bot_{456}$. 

\subsection{Payment Schemes as Smart Contracts} 

We now explain informally how a payment scheme can be deployed in practice as a smart contract running on a blockchain, see \cref{fig:deposit_scheme} for a depiction. Note that there are many subtleties in getting this rigorous, see e.g. \cite{fbc,kachina} for more formal cryptographic modeling. At a high level, we want to ensure party $P_i$ loses $\lambda_{ik}$ utility when the symbol $\sigma_k$ is observed. We can implement this by defining $\lambda_i^* := \max_{k \in \{1,2,\ldots s\}} \lambda_{ik}$ and letting each party $P_i$ make a deposit of size $\lambda_i^*$ to a smart contract before playing the game. Afterwards, the parties are repaid appropriately by the payment scheme to ensure their utility is as dictated by $\L$. Suppose we fix some payment scheme $\Gamma$, then the \emph{augmented game $\Gamma(G)$} is played as follows:
\begin{enumerate}
	\item Each $P_i$ makes a deposit of $\lambda_i^*$ to the payment scheme.
	\item The game $G$ is played, reaching a leaf $\ell_j$.
	\item A symbol $\sigma_k$ is sampled from $\Sigma$ according to the pdf in the $j^\text{th}$ column of $\P$.
	\item Each party $P_i$ is repaid $\lambda_i^* - \lambda_{ik}$.
\end{enumerate}

\begin{figure}
    \centering
    \scalebox{0.725}{
        \begin{tikzpicture}
        \node at (-5,1) {(1)};
        \node[draw=black,circle] (P1) at (-4.95,-2) {$P_1$};
        \node[draw=black,circle] (P2) at (-3.25,-2) {$P_2$};
        \node[] (dots) at (-2.25,-2) {$\cdots$};
        \node[draw=black,circle] (Pn) at (-1.25,-2) {$P_n$};
        \node[draw=black] (pi) at (-2.25,1) {$G$};
        \node (gear) at (0,0) {$\Gamma$};
    	\draw \gear{10}{0.32}{0.4}{8}{2};
        \draw[->] (Pn) -| node[shift={(-0.3,0.3)}] {$\lambda_n^*$} (-0.2,-0.35);
        \draw[-] (P2) |- node[shift={(0.75,0.3)}] {$\lambda_2^*$} (0,-2.75);
        \draw[->] (0,-2.75) -- (0,-0.385);
        \draw[-] (P1) |- node[shift={(0.75,0.3)}] {$\lambda_1^*$} (0.2,-3);
        \draw[->] (0.2,-3) -- (0.2,-0.35);
        \end{tikzpicture}
    }\hspace{1cm}
    \scalebox{0.725}{
        \begin{tikzpicture}
        \node at (-5,1) {(2)};
        \node[draw=black,circle] (P1) at (-4.95,-2) {$P_1$};
        \node[draw=black,circle] (P2) at (-3.25,-2) {$P_2$};
        \node[] (dots) at (-2.25,-2) {$\cdots$};
        \node[draw=black,circle] (Pn) at (-1.25,-2) {$P_n$};
        \node[draw=black] (pi) at (-2.25,1) {$G$};
        \node (gear) at (0,0) {$\Gamma$};
    	\draw \gear{10}{0.32}{0.4}{8}{2};
    	\draw[dashed] (P1) -- (pi);
    	\draw[dashed] (P2) -- (pi);
    	\draw[dashed] (Pn) -- (pi);
        \node[draw=white] at (0,-2.85) {};
        \draw[->] (pi) -|node[shift={(-1.25,0.25)}]{$\sigma_k$} (0,0.425);
        \draw [rounded corners=5,dashed] (-5.6,-2.5) rectangle ++(5,1);
        \end{tikzpicture}
    }\\[0.55cm]
    \scalebox{0.725}{
        \begin{tikzpicture}
        \node at (-6,1) {(3)};
        \node[draw=black,circle] (P1) at (-4.95,-2) {$P_1$};
        \node[draw=black,circle] (P2) at (-3.25,-2) {$P_2$};
        \node[] (dots) at (-2.25,-2) {$\cdots$};
        \node[draw=black,circle] (Pn) at (-1.25,-2) {$P_n$};
        \node[draw=black] (pi) at (-2.25,1) {$G$};
        \node (gear) at (0,0) {$\Gamma$};
    	\draw \gear{10}{0.32}{0.4}{8}{2};
        \draw[->] (-0.38,0.15) -|node[fill=white,shift={(-0.735,-1)}]{$\lambda_1^*-\lambda_{1k}$} (P1);
        \draw[->] (-0.4,0) -|node[fill=white,shift={(-0.735,-1)}]{$\lambda_2^*-\lambda_{2k}$} (P2);
        \draw[->] (-0.38,-0.15) -|node[fill=white,shift={(-0.76,-1)}]{$\lambda_n^*-\lambda_{nk}$} (Pn);
        \end{tikzpicture}
    }
    \caption{Illustration of how the payment scheme $\Gamma$ augments a game $G$. First, in (1), parties make a deposit to $\Gamma$. Then, in (2), the game is run as usual, and a symbol $\sigma_k \in \Sigma$ is emitted to $\Gamma$. Finally, in (3), $\Gamma$ repays the parties based on the symbol it received from $G$.}
    \label{fig:deposit_scheme}
\end{figure}
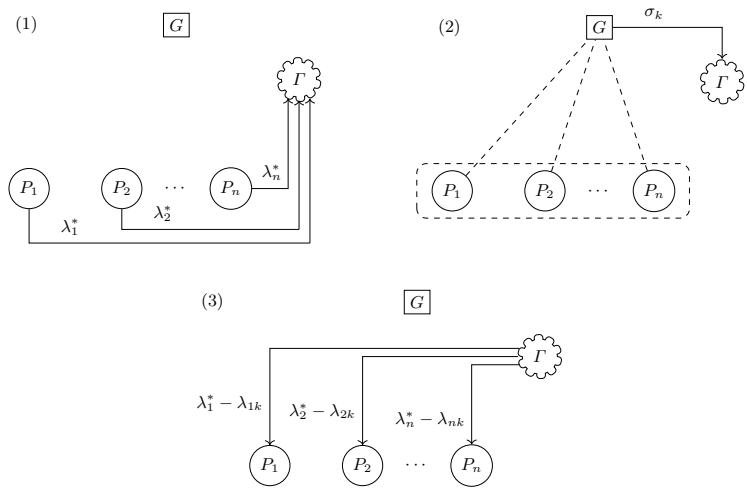

This can be implemented in a fairly straightforward manner using a scripting language and deployed as a smart contract running on a blockchain, assuming access to some information structure with known bounds on the emission probabilities. We stress that this is only one possible implementation of a payment scheme suitable in any scenario, even over the internet when parties are anonymous. The important thing is that party $P_i$ loses $\lambda_{ik}$ utility when symbol $\sigma_k$ is observed. When parties are not anonymous and can be held accountable, the payment scheme can be used in an optimistic manner as was argued by George and Kamara.

\subsection{Game-Theoretic Security}
We now define what it means for a game to be secure in a game-theoretic sense. Intuitively, security should mean the honest strategy profile is an equilibrium, though this is likely not sufficient for some applications. The fact that the honest strategy profile is an equilibrium does not mean it is the only equilibrium. Namely, there might be several dishonest strategy profiles with the same properties, and there is no compelling argument for why parties should opt to be honest in the face of ambiguity. In fact, there might be reasons for being dishonest that are not captured by the utilities of the game, say for spite or for revenge. To remedy this, we want to quantify how much utility parties lose by deviating from the honest strategy profile, in effect measuring the cost of dishonesty. We introduce a parameter $\delta$ such that being dishonest results in the deviating parties losing at least $\delta$ utility. A game with this property is considered secure against $\delta$-deviating rational parties. We give the definition by Schwartzbach \cite{schwartzbach} that generalizes $t$-robust subgame perfect equilibria for finite games of perfect information.

Let $G$ be a fixed finite game with $n$ players and $m$ leaves, and let $\U \in \mathbb{R}^{n \times m}$ be the corresponding utility matrix, and let $\langle \Sigma, \P \rangle$ be some fixed information structure on $G$. We say a utility vector $\mathbf{u}$ is \emph{$C$-inducible in $G$} for a coalition $C$ if there is a strategy $s_C$ such that playing $s=(s_C, s_{-C}^*)$ terminates in a leaf $\ell$ labelled by $\mathbf{u}$ with non-zero probability. 

\begin{definition}
Let $G$ be a game, and $s^*$ an intended strategy profile. We say $G$ has \emph{$\delta$-strong $t$-robust game-theoretic security} if for every subgame of $G$, and every $C$-inducible vector $\mathbf{u}$ in that subgame with $|C| \leq t$, and every $i \in C$, it holds that:
\begin{equation}
	u_i(s^*) \geq \mathbf{u}_i + \delta
\end{equation}
In other words, every coalition of $\leq t$ parties that deviates from $s^*$ at any point in the game should lose at least $\delta$ utility for each deviating party. We note that for finite games of perfect information, $t$-robust subgame perfect equilibria is retained as a special case of this definition by letting $\delta = 0$. 
\end{definition}

\subsection{Example: Decentralized Commerce}

In this section, we demonstrate the applicability of our model by considering the problem of \emph{decentralized commerce} (the `eBay problem'). Here, a seller $S$ wants to sell an item \emph{it} over the internet to a buyer $B$ for $x$ money. To make the problem non-trivial, we assume \emph{it} is physical such that the protocol cannot be entirely implemented using cryptography (see e.g. \cite{optimistic_fair_exchange,multiple_arbiter_fair_exchange,usable_optimistic_fair_exchange,fairswap,dual_deposit} for solutions that work with digital goods). We assume \emph{it} has a value of $y$ to the buyer, and a value of $x'$ to the buyer. To make the problem feasible, we assume that $y > x > x' > 0$. We consider a simple game where $S$ first decides whether to send \emph{it} to $B$, after which $B$ decides whether or not to pay $S$. The resulting extensive-form game is depicted in \cref{fig:commerce}.

In this simple game, the trade will never be completed, as it is evidently rational for the buyer to always reject delivery of the item; consequently, it is rational for the seller not to send the item. This seems to contradict empirical data, as variants of this game are played successfully all the time. The reason for this is that, in practice, buyer and seller are not anonymous and can be held accountable for fraud, and potentially subject to legal repercussions. Also, such trades are typically processed by a middleman (such as eBay/Amazon/Alibaba) that may offer some insurance for either buyer or seller. Regardless, our goal is to augment the game with a payment scheme to avoid having to trust a middleman and enable fully decentralized commerce.

\begin{figure}
	\centering
	\begin{tikzpicture}
	[level 1/.style={sibling distance=35mm},
	level 2/.style={sibling distance=20mm}, 
	level distance=1.5cm,align=center,
    every node/.style={thin},
    emph/.style={edge from parent/.style={ very thick,draw}},
    norm/.style={edge from parent/.style={solid,black,thin,draw}}]
	\node[draw,circle,norm] {$S$}
	child[emph] {
		node[draw,circle] {$B$}
		child[norm] { 
		    node {$\begin{array}{l} B:\,-x\\S:\,\phantom{-}x\end{array}$} 
		    edge from parent node[anchor=south east] {accept}}
		child { 
		    node {$\begin{array}{l}B:\,0\\S:\,0\end{array}$} 
    		edge from parent node[anchor=south west] {\bf reject}}
		edge from parent node[anchor=south east] {\bf not send}
	}
	child[norm] {
		node[draw,circle] {$B$}
		child[emph] { 
		    node {$\begin{array}{l} B:\,\phantom{-}y\phantom{'}\\S:\,-x' \end{array}$} 
		    edge from parent node[anchor=south east] {\bf reject}}
		child { 
		    node {$\begin{array}{l} B:\,y-x\phantom{'}\\S:\,x-x'\end{array}$}
		    edge from parent node[anchor=south west] {accept} }
		edge from parent node[anchor=south west] {send}
	};
	\end{tikzpicture}
	\caption{Extensive-form representation of the decentralized commerce game. The dominating paths are shown in bold. We observe that the dominating strategy is for the seller not to send the item, and for the buyer to withhold their payment regardless of whether they received the item.}
	\label{fig:commerce}
\end{figure}
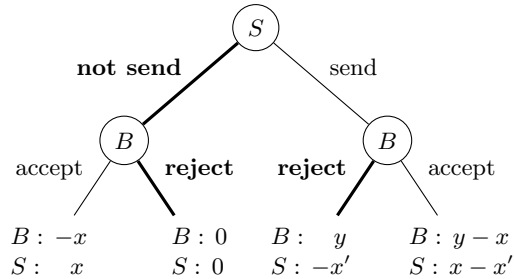

To do so, let us first assume the parties use a smart contract to process the trade, and let us assume we have \emph{some} means of probabilistically inferring whether the item was actually shipped. In the blockchain literature, such a mechanism is called a \emph{blockchain oracle} \cite{blockchain_oracle}. In particular, we assume the oracle can recover the ground truth with probability $1-\varepsilon$ for some constant $\varepsilon<\frac12$. One potential blockchain oracle is the jury-based system Kleros \cite{kleros}, that was shown in \cite{adjudication} to yield $\varepsilon<\frac12$ under reasonable assumptions on the jurors. To proceed, we need to define an information structure on the game. We first define an alphabet of outcomes as follows.
$$
    \Sigma = \{\top, \bot_B, \bot_S\}
$$
Here $\top$ is a symbol emitted if the buyer accepts the trade, and $\bot_B, \bot_S$ are outcomes of the oracle if it is invoked, where $\bot_B$ (resp. $\bot_S$) means `the buyer (resp. the seller) was dishonest'. We then define an emissions matrix on the game as follows.
$$
    \P =    \begin{pmatrix}
                1 & 0 & 0 & 1\\
                0 & 1-\varepsilon & \varepsilon & 0\\
                0 & \varepsilon & 1-\varepsilon & 0
            \end{pmatrix}
$$
Now, let us assume we want to instantiate payments to ensure the game has $x$-strong game-theoretic security. We proceed using backward induction in \cref{fig:commerce}, letting the corresponding utilities have a difference of $\geq x$. This is e.g. achieved by defining the following `desired' utility matrix $\mathbf{E}$. 
$$
    \mathbf{E} =    \begin{pmatrix}
                        -x & 0 & y-2x & y-x\\
                        x  & -x' & -x' & x-x'
                    \end{pmatrix}
$$
In order to implement $\mathbf{E}$, the payment scheme $\L$ must satisfy \cref{eq:utilities} as follows.
\begin{align*}
    \L \P &= \U - \mathbf{E} \\&\Longleftrightarrow 
            \begin{pmatrix}
                \lambda_{B\top} & (1-\varepsilon)\lambda_{B\bot_B} + \varepsilon \lambda_{B\bot_S} & \varepsilon\lambda_{B\bot_B} + (1-\varepsilon) \lambda_{B\bot_S} & \lambda_{B\top}\\
                \lambda_{S\top} & (1-\varepsilon)\lambda_{S\bot_B} + \varepsilon \lambda_{S\bot_S} & \varepsilon\lambda_{S\bot_B} + (1-\varepsilon) \lambda_{S\bot_S} & \lambda_{S\top}
            \end{pmatrix}
          = \begin{pmatrix}
                0 & 0 & 2x & 0\\
                0 & x' & 0 & 0
            \end{pmatrix}
\end{align*}
This immediately gives $\lambda_{B\top} = \lambda_{S\top} = 0$, while the remaining payments are given by four equations with four unknown and can be solved using Gaussian elimination to yield the following payment scheme.
$$
    \L = 
    \begin{pmatrix}
        0 & -\frac{2 \varepsilon}{1-2\varepsilon}x & \phantom{-}\frac{2(1-\varepsilon)}{1-2\varepsilon}x\\
        0 & \phantom{-}\frac{1-\varepsilon}{1-2\varepsilon} x' & -\frac{\varepsilon}{1-2\varepsilon} x'
    \end{pmatrix}
$$
In other words, the buyer must make a deposit of size $\frac{2(1-\varepsilon)}{1-2\varepsilon} x$ to the smart contract, while the seller must make a deposit of size $\frac{1-\varepsilon}{1-2\varepsilon} x'$. To make this more concrete, suppose we have $x = 100\$, x' = 50\$$, and $\varepsilon=0.1$. Then the buyer must make a deposit of size $\lambda_B^* = 225\$$, while the seller must make a deposit of size $\lambda_S^* = 225/4 \$ \approx 57 \$$. The large difference in the deposits reflects the fact that the protocol is `biased' in favor of the buyer. In practice, while $x$ is known to the mechanism, $x'$ is usually not. Instead, we can use a variation of this payment scheme proposed by Schwartzbach \cite{schwartzbach}. Here, both buyer and seller submit a deposit of size $\lambda_B^*=\lambda_S^* = x$, and the resulting contract is shown to have $(1-2\varepsilon)x$-strong game-theoretic security. Their contract is also \emph{optimistic} in the sense that deposits are only used when the buyer disputes delivery of the item.

\section{Computational Complexity}
\label{complexity}
In this section, we analyze the computational complexity of finding payment schemes in arbitrary games. For games of perfect information, we observe the problem is equivalent to linear programming (denoted \lp) under logspace-reductions, thus showing the problem is complete for \textsf{P}. 

More formally, we consider the following optimization problem.

\prob{$\textsf{PaymentScheme}_t^\delta$}{Finite game $G$ with utility matrix $\U \in \mathbb{R}^{n \times m}$ and intended strategy profile $s^*$; finite alphabet $\Sigma$ with $s=|\Sigma|$, emission matrix $\P \in \mathbb{R}^{s \times m}$, and cost vector $\mathbf{c} \in (\mathbb{R}\cup \{\infty\})^{ns}$.}{Self-contained payment scheme $\L \in \mathbb{R}^{n\times s}$ s.t. $\Gamma^\L(G)$ has $\delta$-strong, $t$-robust game-theoretic security, for which $\mathbf{c}^\top \text{vec}(\L)$ is minimized; or $\bot$ if no such payment scheme exists.}
Here, $\infty$ is a formal symbol in the cost function that `forces' the corresponding payment to equal zero. It does not contribute to the actual cost function. This can e.g. be used to implement \emph{honest invariance}, to ensure the utility vector for the intended strategy remains unchanged.  We allow this modeling to simplify our reductions, though we can make do without this assumption; we sketch how to do so at the end of the section.

\begin{theorem}
    $\ps$ is \textsf{P}-complete for games of perfect information.
\end{theorem}
We prove this in the next two subsections, by reducing both to and from $\lp$ using logspace-reductions. For games of imperfect information, it is unlikely we can find an optimal payment scheme to change the equilibrium, as even computing the equilibrium for these games is known to be \textsf{PPAD}-complete. As a result, we conjecture the problem to be hard.
\begin{conjecture}
    $\ps$ is \textsf{PPAD}-hard for imperfect information games.
\end{conjecture}

\subsection{Upper Bound: Reduction to LP} 
In this section, we show how to reduce $\ps$ to \lp. Since the feasible region is a convex polyhedron, it is unsurprising that we can use linear programming to decide the minimal size of the deposits necessary to establish security. In particular, we can write the necessary constraints for $\delta$-strong $t$-robust game-theoretic security as a set of linear constraints. For convenience, we will represent the utility matrix $\U$ as a vector $\mathbf{u} \in \mathbb{R}^{nm}$ in row-major order. We will then collect the set of necessary constraints in a matrix $\AA{t} \in \mathbb{R}^{\alpha^{(t)}\times nm}$ where $\alpha^{(t)}$ denotes the number of such constraints. We also let $\e{t}=[\delta,\delta,\ldots,\delta]^\top \in \mathbb{R}^{\alpha^{(t)}}$ be a vector only containing $\delta$. Note that $\alpha^{(t)}$ is a constant that depends on the structure of the game.
\begin{proposition}
    $\ps$ can be reduced to $\lp$ in logspace.
\end{proposition}
\begin{proof}
First note that the set of utility matrices with $\delta$-strong $t$-robust game-theoretic security can be recovered as the set of solutions to the following equation:
\begin{equation}\label{eq:aae}
\AA{t} \mathbf{v} \geq \e{t}
\end{equation}
We note that, in general, it is hard to give an exact expression for $\AA{t}$ since this is tightly dependent on the structure of the game, as is $\alpha^{(t)}$. Instead, the matrix $\AA{t}$ can be computed using a simple recursive procedure. In the base case, the leaves, there are no constraints. At each branch owned by a player $P_i$, we need to bound the probability of each undesirable outcome in terms of the honest outcome $\mathbf{u}^*$. To do so, we compute the \emph{$t$-inducible region}, defined as the set of outcomes inducible by a coalition $C$ containing $P_i$ of size $\leq t$. For each outcome $\mathbf{v}$ in the $t$-inducible region, we add a column $\boldsymbol{\psi} \in \mathbb{R}^{nm}$ to $\AA{t}$ that ensures that $\mathbf{u}^*_i \geq \mathbf{v}_i + \delta$. To do so, suppose $\mathbf{u}^*_i$ and $\mathbf{v}_i$ have indices $a,b$ respectively, we then let $\boldsymbol{\psi}_{im+a} \gets 1$, and $\boldsymbol{\psi}_{im+b} \gets -1$ and zero elsewhere, and add an entry containing $\delta$ to $\e{t}$. This procedure can be completed using a single pass of the game tree by keeping track of the $t$-inducible region as we go along. Note that there is a technical issue since our decision variables $\L$ are not in vector form, as is usual of linear programming. To remedy this, we also want to collect the deposits in a vector $\boldsymbol{\lambda} \in \mathbb{R}^{ns}$ in row-major order. For a given information structure $\langle \Sigma, \P\rangle$, we construct a matrix $\mathbf{R}$ equivalent to $\P$ in the following way: for every index $ij$ in $\L$, we construct the `base matrix' $\mathbf{L}^{ij}$ that is 1 in index $ij$, and 0 everywhere else. We then compute a row of $\mathbf{R}$ by computing the product $\mathbf{L}^{ij} \P$ and putting it in row-major order. It is not hard to see that the image of $\P$ is isomorphic to the column space of $\mathbf{R}$, and hence we say $\boldsymbol{\lambda}$ implements the utility vector $\mathbf{e}$ iff $\mathbf{e} = \mathbf{u} - \mathbf{R}\boldsymbol{\lambda}$. We now substitute this in \cref{eq:aae} to get 
$
\AA{t}(\mathbf{u}-\mathbf{R}\boldsymbol{\lambda}) \geq \e{t}
$. Next, we move the constant terms to the right-hand side to yield the following:
\begin{equation}\label{eq:aae3}
-\AA{t}\mathbf{R}\boldsymbol{\lambda} \geq \e{t}-\AA{t}\mathbf{u}
\end{equation}
We also have to ensure the payment scheme is self-contained, but this is a simple set of linear constraints $\sum_{i=1}^n \lambda_{is+k}$ for every $k=1\ldots s$. Finally, note that our objective function is $\mathbf{c}^\top \boldsymbol\lambda$, and since all constraints are linear we can produce the following linear program.
\begin{align*}
	\text{min}\quad&\mathbf{c}^\top \boldsymbol{\lambda}\\
	\text{s.t.}\quad
	&-\AA{t}\mathbf{R}\boldsymbol{\lambda} \geq \e{t}-\AA{t}\mathbf{u}\\
	&\sum_{i=1}^n \lambda_{is+k} \geq 0&& \forall\, k=1\ldots s
\end{align*}
To deal with $\infty$ in the cost function, we may set the corresponding cost of the linear program to an arbitrary value and add an equality constraint to ensure the decision variable equals zero. Note that the linear program can be constructed by maintaining a constant set of pointers to the game given as input, which concludes the proof.\qed
\end{proof}

\subsection{Lower Bound: Reduction from LP}
We now show how to reduce \lp\, to $\textsf{PaymentScheme}^0_1$ using logarithmic space. The resulting game is a two-player finite game of perfect information. The reduction can easily be adapted to any $\delta\geq 0, t\geq1$. Consider an arbitrary instance of \lp, $\{\min \mathbf{c}^\top \mathbf{x} \mid \mathbf{A}\mathbf{x} \geq \mathbf{b}, \mathbf{x}\geq \mathbf{0}\}$,
where $\mathbf{c}=(\mathbf{c}_i) \in \mathbb{R}^n, \mathbf{A}=(\mathbf{a}_{ij}) \in \mathbb{R}^{m \times n}$, and $\mathbf{b}=(\mathbf{b}_i) \in \mathbb{R}^m$. Without loss of generality, we will assume that the columns of $\mathbf{A}$ have a positive column sum. This can be achieved by shifting $\mathbf{A}$ and $\mathbf{b}$ correspondingly.

\begin{proposition}
    $\lp$ can be reduced to $\textsf{PaymentScheme}^0_1$ in logspace.
\end{proposition}
\begin{proof}
At a high level, the reduction is as follows. We first describe the game, and afterwards derive a suitable information structure. The game consists of two players $P_1,P_2$. The root of the game consists of a move for player $P_1$ who wants to `sabotage' satisfaction of the program. They get utility 1 if they sabotage an inequality, and 0 otherwise. They are allowed to choose between a set of $m$ gadgets, one for each inequality $\mathbf{a}_i^\top \mathbf{x} \geq \mathbf{b}_i$. In addition, they can choose a `target' leaf that pays 0 to both players. Each gadget consists of a move for the other player $P_2$ who can choose whether to satisfy their inequality or not. If they sabotage their inequality (move `left') they get $\mathbf{b}_i$ utility, otherwise if they move `right' they get 0 utility. See \cref{fig:reduction_to_ps} for an illustration. Clearly, the SPE of the game is for $P_2$ to move left in the $i^\text{th}$ gagdet if $\mathbf{b}_i>0$, and for $P_1$ to choose any convex combination of the gadgets for which the players move left. Our goal is to design an information structure for which a payment scheme can ensure that $P_1$ chooses the target if and only if all inequalities are satisfied.

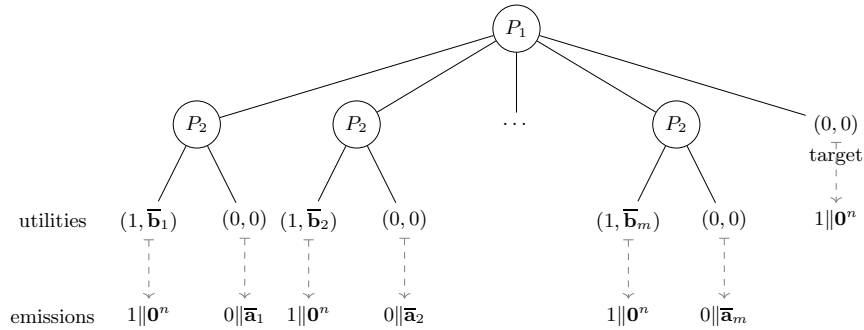
\begin{figure}
    \centering
    \scalebox{0.85}{
    \begin{tikzpicture}
	[level 1/.style={sibling distance=25mm},
	level 2/.style={sibling distance=15mm}, 
	level distance=1.5cm,align=center,
    every node/.style={thin}]
    \node at (-7.75,-3)  {utilities};
    \node at (4.5,-2)  {target};
    \node at (-7.75,-4.5)  {emissions};
	\node[draw,circle,xshift=-0.5cm] {$P_1$}
	child {
		node[draw,circle] {$P_2$}
		child { 
		    node {$(1,\overline{\mathbf{b}}_1)$}
		    child[edge from parent/.style={|->,dashed,gray,draw}] {
		        node {$1 \Vert \mathbf{0}^n$}
		    }
		}
		child { 
		    node {$(0,0)$} 
		    child[edge from parent/.style={|->,dashed,gray,draw}] {
		        node {$0 \Vert \overline{\mathbf{a}}_1$}
		    }
		}
	}
	child {
		node[draw,circle] {$P_2$}
		child { 
		    node {$(1,\overline{\mathbf{b}}_2)$} 
		    child[edge from parent/.style={|->,dashed,gray,draw}] {
		        node {$1 \Vert \mathbf{0}^n$}
		    }
		}
		child { 
		    node {$(0,0)$} 
		    child[edge from parent/.style={|->,dashed,gray,draw}] {
		        node {$0 \Vert \overline{\mathbf{a}}_2$}
		    }
		}
	}
	child{node{$\cdots$}}
	child {
		node[draw,circle] {$P_2$}
		child { 
		    node {$(1,\overline{\mathbf{b}}_m)$} 
		    child[edge from parent/.style={|->,dashed,gray,draw}] {
		        node {$1 \Vert \mathbf{0}^n$}
		    }
		}
		child { 
		    node {$(0,0)$}
		    child[edge from parent/.style={|->,dashed,gray,draw}] {
		        node {$0 \Vert \overline{\mathbf{a}}_m$}
		    }
		}
	}
	child { 
	    node {$(0,0)$} 
	    child[edge from parent/.style={|->,dashed,gray,draw}] {
	        node {$1 \Vert \mathbf{0}^n$}
	    }
    };
	\end{tikzpicture}}
    \caption{Depiction of the reduction from $\lp$ to $\textsf{PaymentScheme}^0_1$. The dashed arrows depict the corresponding information structure (pdf for each leaf). The player $P_1$ wants to sabotage satisfiability of the circuit and gains 1 utility for doing so (0 otherwise). The player $P_2$ will sabotage the $i^\textrm{th}$ gadget (and hence allow $P_1$ to win) if and only if the $i^\textrm{th}$ inequality is not satisfied. A payment scheme corresponds to an assignment of the variables in the $\lp$-instance, with emission probabilities proportional to the weights, such that an equilibrium with the target in its support corresponds to a satisfying assignment of the variables. }
    \label{fig:reduction_to_ps}
\end{figure}

We now describe the information structure of the game. We have to specify an alphabet and a pdf for each leaf of the game. We will have $\Sigma=\{\top, \bot_1, \bot_2, \ldots \bot_n\}$, where $\top$ means `all inequalities are satisfied', while $\bot_i$ is associated with the decision variable $\mathbf{x}_i$. When $P_2$ satisfies their inequality, the symbol $\top$ is outputted with probability 1. When $P_2$ sabotages their inequality, the column $\mathbf{a}_i$ is used a pdf to sample the symbols $\{\mathbf{x}_i\}_{i=1}^n$. Of course, $\mathbf{a}_i$ is not necessarily a pdf, but we can normalize it by defining $\overline{\mathbf{a}}_{ij} = \frac{\mathbf{a}_{ij}}{\sum_{k=1}^n \mathbf{a}_{ik}}$. We similarly define $\overline{\mathbf{b}}_i = \frac{\mathbf{b}_i}{\sum_{k=1}^n \mathbf{a}_{ik}}$ and use $\overline{\mathbf{b}}_i$ in lieu of $\mathbf{b}_i$ in the gadgets. This operation is well-defined since $\mathbf{A}$ was assumed to have positive column sums, and inequalities are preserved under positive scaling. To summarize, when a player goes left, the corresponding pdf is $[1,0,0,\ldots,0]^\top\in \mathbb{R}^{n+1}$, and when a player goes right, the pdf is $0\Vert\overline{\mathbf{a}} \in \mathbb{R}^{n+1}$. As the intended strategy profile $s^*$, we consider any strategy profile where $P_2$ always move right and $P_1$ chooses an arbitrary gadget. The cost function $\hat{\mathbf{c}}$ of the payment scheme will be defined as follows,
$$
    \hat{\mathbf{c}} := [\overbrace{\infty,\infty,\ldots, \infty}^\text{$n+1$ terms},\infty,\mathbf{c}_1,\mathbf{c}_2,\ldots,\mathbf{c}_n]^\top \in \mathbb{R}^{2n+2},
$$
Now, suppose $\L\in\mathbb{R}^{(n+1)\times(m+1)}$ is output as an optimal payment scheme. Let $\L_{\bullet i}$ denote the $i^\textrm{th}$ \emph{row} of $\L$ (as a column vector) corresponding to the leaf where $P_i$ goes right, and let $\P_i=(0 \Vert \overline{\mathbf{a}}_i)$ be the column of $\P$ corresponding to going right in the $i^\textrm{th}$ gadget. Now, since some of the weights are $\infty$, we know that $\L_1=\mathbf{0}$ and $\L_{\bullet 1}=\mathbf{0}$. Hence, the utility vector going right remains $[1,0,0,\ldots,0]^\top$ for each gadget, and the utility for $P_1$ remains unchanged. By optimality and since $\delta=0, t=1$, we know that $s^*$ must an SPE. This means that $P_2$ must receive (at least) as much utility going left as they do going right (otherwise $P_1$ would not hit the target).  Then by \cref{eq:utilities}, we must have,
$$
    \forall i.\,(-\P_i^\top \L_{\bullet 2} \geq \overline{\mathbf{b}}_i) \quad \Longleftrightarrow \quad \forall i.((0 \Vert \overline{\mathbf{a}}_i)^\top (-\L_{\bullet 2}) \geq \overline{\mathbf{b}}_i \quad \Longleftrightarrow\quad \mathbf{A} \mathbf{x} \geq \mathbf{b}
$$
where $\mathbf{x} := [\L_{i2}]_{i=1}^m$ is the vector consisting of the non-zero (last $m$) entries of $\L_{\bullet 2}$. This means that $s^*$ is an SPE if and only if the inequalities are satisfies. We know further that $\mathbf{x}\geq \mathbf{0}$ since the payment scheme is self-contained. Minimization of the objective function $\mathbf{c}^\top \mathbf{x}$ comes directly from minimization of $\hat{\mathbf{c}}^\top \text{vec}(\L)$, as some of the weights are $\infty$. Finally, note that all parts of the reduction can be performed by maintaining a constant set of pointers, thus concluding the proof. \qed
\end{proof}

\paragraph{Removing $\infty$.} To remove $\infty$ from the optimization problem, we may add an additional dummy player $P_3$ who provides the necessary `liquidity' to $P_2$, while ensuring $P_1$'s utility is left unchanged (note that the payment scheme must be self-contained, i.e. column sums of $\L$ must be non-negative). We assign to $P_3$ arbitrary utilities in the reduction, and assign to the payments of $P_3$ the opposite weights given to $P_2$, i.e. $-\mathbf{c}_i$ instead of $\mathbf{c}_i$. The weights given to the payments of $P_1$ are all zero. It is not hard to see that the resulting payment scheme has the same set of optimal values, as optimization problems are invariant under scaling. In addition, all payments to $P_1$ must be zero as any solution with non-zero payments to $P_1$ are strictly dominated by assigning the payment to either $P_2$ or $P_3$ if the corresponding weights are non-zero. If instead, the corresponding weights of $P_2,P_3$ are both zero, we can slightly perturb the cost of $P_1$ to e.g. 1 to ensure the utility of $P_1$ is unchanged.

\section{Case Study: Secure Rational MPC from PVC}
\label{mpc}
In this section, we apply our framework to a more complicated scenario involving secure multiparty computation (MPC). Our work is similar to \cite{fbc}, in that we also use payments to incentivize honesty from a PVC protocol. However, they focus mainly on the cryptographic modeling, while our focus is mainly game-theoretic and thus complements their work. We start with a brief and informal definition of MPC for the purpose of self-containment, and refer to \cite{mpc_book} for more details and formal definitions. 

\paragraph{Secure Multiparty Computation (MPC).} In MPC, a set of $n$ mutually distrusting parties $P_1,P_2,\ldots,P_n$ want to compute a public function $f$ on their private data $\mathbf{x}=(x_1,x_2,\ldots,x_n)$. The parties engage in an interactive protocol that ends with each of them producing an output $y_i$. The goal is for the output to be \emph{correct} such that $y_i=f(\mathbf{x})$, and \emph{private}, meaning the protocol leaks no information about the inputs of the parties, other than that which can be gathered from the function output itself. This should hold even if a coalition of $t$ parties are controlled by a monolithic adversary who tries to break security of the protocol. MPC is a large research area with many proposed protocols, depending on the assumptions. One of the weakest notions of security is that of \emph{passive security} where correctness and privacy are guaranteed against an `honest-but-curious' adversary, who adheres honestly to the protocol description but tries to collect more information than they should. Such protocols are typically comparatively cheap, in contrast to protocols with \emph{active security} that remain secure even if the adversary may deviate arbitrarily from the protocol description. Active protocols are typically orders of magnitude more expensive than their passive counterparts. To remedy this, Aumann and Lindell \cite{covert} propose an intermediate notion of security called \emph{covert security} where the adversary is allowed to cheat, but is caught with some constant non-zero probability. They propose three different definitions, giving different power to the adversary. The weakest notion is `failed simulation' where the adversary learns the inputs of the honest parties when caught, while the strongest is called `strong explicit cheat formulation' where they do not. In the present section, we opt for the latter, though our model easily adapts to the former albeit with larger payments. A disadvantage of covert secure protocols is that they do not allow the participants to convince a third party who was dishonest which means they are not directly applicable to our setting. This was augmented to \emph{publicly verifiable covert security} (PVC) by Asharov and Orlandi \cite{pvc} where a proof of cheating is output that can be verified by a third party. The underlying assumption of these protocols is that the adversary suffers some cost from being caught, meaning it is rational for them not to cheat. The typical use-case is that of competing businesses who may wish to perform some joint computation on trade secrets but are not willing to risk tarnishing their name. While this may be a reasonable assumption in many cases, it is unclear that this works in e.g. an anonymous setting where the parties cannot be held accountable. Instead, we will use a payment scheme to prove it is rational for the parties not to cheat. Our plan is to analyze the information structure induced by the definition of covert security. We then apply our payment schemes to the resulting game and derive values for the deposits of the parties. 
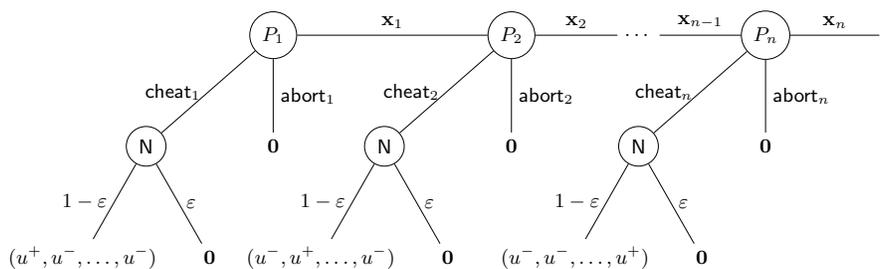
\begin{figure}
	\centering
	\resizebox{\linewidth}{!}{
		\begin{tikzpicture}
		[
		level 1/.style={sibling distance=20mm},
		level 2/.style={sibling distance=20mm},
		level 3/.style={sibling distance=20mm},
		level distance=1.75cm,align=center]
		\node[draw,circle] {$P_1$}
		child {
			node[draw,circle] {$\textsf{N}$}
			child{
				node {$(u^+,u^-,\ldots, u^-)$}
				edge from parent node [left] {$1-\varepsilon$}
			}
			child{
				node {$\mathbf{0}$}
				edge from parent node [right] {$\varepsilon$}
			}
			edge from parent node [left] {$\textsf{cheat}_1$}
		}
		child {
			node {$\mathbf{0}$}
			edge from parent node [right] {$\textsf{abort}_1$}
		}
		child [xshift=1.75cm,yshift=1.75cm]{
			node[draw,circle] {$P_2$}
			child {
				node[draw,circle] {$\textsf{N}$}
				child{
					node {$(u^-,u^+,\ldots, u^-)$}
					edge from parent node [left] {$1-\varepsilon$}
				}
				child{
					node {$\mathbf{0}$}
					edge from parent node [right] {$\varepsilon$}
				}
				edge from parent node [left] {$\textsf{cheat}_2$}
			}
			child {
				node {$\mathbf{0}$}
				edge from parent node [right] {$\textsf{abort}_2$}
			}
			child {
				node[yshift=1.75cm] {$\dots$}
				child[xshift=2cm,yshift=1.75cm]{
					node[draw,circle] {$P_n$}
					child[yshift=1.75cm,xshift=4cm]{
						node {$\mathbf{1}$}
						edge from parent node [above] {$\mathbf{x}_n$}
					}
					child[xshift=-2cm] {
						node[draw,circle] {$\textsf{N}$}
						child{
							node {$(u^-,u^-,\ldots, u^+)$}
							edge from parent node [left] {$1-\varepsilon$}
						}
						child{
							node {$\mathbf{0}$}
							edge from parent node [right] {$\varepsilon$}
						}
						edge from parent node [left] {$\textsf{cheat}_n$}
					}
					child[xshift=-2cm] {
						node {$\mathbf{0}$}
						edge from parent node [right] {$\textsf{abort}_n$}
					}
					edge from parent node [above] {$\mathbf{x}_{n-1}$}
				}
				edge from parent node [above] {$\mathbf{x}_2$}
			}
			edge from parent node [above] {$\mathbf{x}_1$}
		};
		\end{tikzpicture}
	}
	\caption{
		The ideal functionality $\Fpvc$ with the strong explicit cheat formulation represented as an extensive-form game $G_\texttt{PVC}$ rooted at $P_1$. Our goal is to augment the functionality with a payment scheme such that it is rational to behave honestly. Note that in this representation, for clarity there are two distinct leaves when a player $P_i$ attempts cheating, though in the following we `merge' the two nodes belonging to nature for simplicity. 
	}
	\label{fig:ideal_game}
\end{figure}

\paragraph{Secure Rational MPC from PVC.} We consider a set of $n$ parties $P_1, P_2, \ldots, P_n$ interacting with the ideal PVC functionality $\mathcal{F}_\texttt{PVC}$. To analyze the interaction using game theory, we need to be able to give some bounds on the utilities of the parties. In order to simplify the presentation, we assume the parties are homogeneous, in that they have the same utility functions. We further disregard the cost of running the protocol, e.g. transaction fees, such that any $\textsf{abort}_i$ gives 0 utility to all parties. Note that we can always normalize the utilities in a game as this preserves the total order. As such, we assume a party receives 1 utility if they send their input and receive back the correct output. If instead a party cheats and is successful, they receive $u^+$ utility, while a party whose input is revealed receives $u^-$ utility.  As $\Fpvc$ does not explicitly punish parties who are caught cheating, we assume a party who is caught cheating receives 0 utility. As in \cite{pvc}, we are using the strong explicit cheat formulation from \cite{covert}, so a cheater who is caught does not learn the inputs of the honest parties, and as such earns 0 utility. We are not modeling the fact that parties can send incorrect inputs, for the simple reason that it is impossible for the payment scheme, in general, to detect this. We assume that parties always send their input truthfully, or rather their true input is defined to be whatever they send to the functionality. The corresponding information structure would not be able to distinguish the two classes of leaves, the distributions would be linearly dependent, making it impossible to instantiate the deposits to ensure security. For some specific applications however, one could imagine a function that allows to determine if a party did provide the wrong input. It is not hard to augment our model to accommodate this scenario, though it is out of scope for the present paper. 

To make the problem nontrivial, we require that
$
	u^+ > 1 > 0 > u^-
$. Note that we are assuming the parties are oblivious to the utility earned by other parties. This is in contrast to \cite{halpern_teague} who assume parties strictly prefer that as few other parties learn the output as possible. This is not to circumvent their impossibility result, as this is accomplished by allowing the deposits to alter the total order of outcomes, i.e. we assume quasilinearity. Rather, it is for simplicity of exposition, though it would be interesting as future work to augment our model to this setting. We represent the interaction as an extensive-form game $G_\texttt{PVC}$, and draw the corresponding tree. W An illustration of the game tree can be found in \cref{fig:ideal_game}. It is not hard to see that when $(1-\varepsilon)\,u^+ > 1$, the only equilibrium in the game is for $P_1$ to attempt to cheat. Instead, we want all players to play honestly. First, we need to define an information structure on the game. We first remark that the structure of the game is such that only one party can deviate in any given strategy profile. This means we can define the following alphabet of possible outcomes as, 
$
\Sigma = \{\top, \textsf{abort}_1, \textsf{cheat}_1, \textsf{abort}_2, \textsf{cheat}_2, \ldots, \textsf{abort}_n,\textsf{cheat}_n\}
$. We assume the symbols are ordered left-to-right. Here $\top$ is a symbol emitted when no cheating was detected, and no aborting occurred. Note that this overloads the notation of $\textsf{abort}_i$ and $\textsf{cheat}_i$. We now analyze the information structure induced by the functionality. For simplicity, we will slightly modify the game tree in \cref{fig:ideal_game}. Namely, we collapse each subgame corresponding to a move by nature into a single leaf with expected utility $(1-\varepsilon)\,u^+$. This allows us to write a single pdf for that leaf. If we instead insist on having separate leaves, then the columns are no longer linearly independent; hence \cref{lemma:Phi_must_be_left_invertible_to_implement_any_E} does not apply directly; however, it still applies if we replace `the inverse' with `a left inverse'. This needlessly complicates the analysis, hence the simplifying assumption. If all parties are honest, we reach the outcome $\mathbf{1}$ and the symbol $\top$ is emitted. If some party $P_i$ aborts, the output of the honest parties will always be $\textsf{abort}_i$. If instead, a party attempts to cheat, with probability $\varepsilon$ they are caught and the message $\textsf{cheat}_i$ is output. If they are not caught, the symbol $\top$ is also emitted. Suppose the leaves of $G_\texttt{PVC}$ are ordered left-to-right in \cref{fig:ideal_game}, then we can write the information structure as follows.

$$
	\mathbf{\Phi}_\texttt{PVC} = \begin{pmatrix}
	0 & 1-\varepsilon & 0 & 1-\varepsilon & \cdots & 0 & 1-\varepsilon & 1 \\
	1 & 0 & 0 & 0 & \cdots& 0&0&0\\
	0 & \varepsilon & 0 & 0 & \cdots & 0&0&0\\
	0 & 0 & 1 & 0 &\cdots & 0&0&0\\
	0 & 0 & 0 & \varepsilon & \cdots & 0&0&0 \\
	\vdots &\vdots&\vdots& \vdots&& \vdots & \vdots & \vdots &  \\
	0 & 0 & 0 & 0 & \cdots & 1 & 0 & 0 \\
	0 & 0 & 0 & 0 & \cdots & 0 & \varepsilon & 0
	\end{pmatrix}
$$

\noindent It is not hard too see that when $\varepsilon>0$, all columns are linearly independent, and as such $\mathbf{\Phi}_\texttt{PVC}$ is invertible. In fact, simple Gaussian elimination implies its inverse is as follows.

$$
	\mathbf{\Phi}_\texttt{PVC}^{-1} = \begin{pmatrix}
		0 & \frac{1}{\varepsilon} & 0 & 0 & 0 & \cdots & 0 & 0\\
		0 & 0 & 1 & 0 & 0 & \cdots & 0 & 0\\
		0 & 0 & 0 & \frac{1}{\varepsilon} & 0 & \cdots & 0 & 0\\
		0 & 0 & 0 & 0 & 1 & \cdots & 0 & 0\\
		\vdots & \vdots & \vdots & \vdots & \vdots & &\vdots & \vdots \\
		0 & 0 & 0 & 0 & 0 & \cdots & \frac{1}{\varepsilon} & 0\\
		0 & 0 & 0 & 0 & 0 & \cdots & 0 & 1\\
		1 & \frac{\varepsilon-1}{\varepsilon} & 0 & \frac{\varepsilon-1}{\varepsilon} & 0  & \cdots & \frac{\varepsilon-1}{\varepsilon} & 0 
	\end{pmatrix}
$$
By \cref{lemma:Phi_must_be_left_invertible_to_implement_any_E}, we can implement any utility matrix $\mathbf{E}$. In order to obtain $(\delta+1)$-strong game-theoretic security we could for instance define the following:
$$
	\mathbf{E} = 
	\begin{pmatrix}
		-\delta & -\delta & 0 & 0 & \cdots & 0 & 0 & 1 \\
		0 & 0 & -\delta & -\delta &  \cdots & 0 & 0 & 1 \\
		\vdots & \vdots & \vdots & \vdots & & \vdots &\vdots & \vdots \\
		0 & 0 & 0&0&\cdots & -\delta & -\delta & 1 \\
		
	\end{pmatrix}
$$

\noindent In this setting, any party who deviates gains an expected utility of $-\delta$, while they gain 1 utility by following the strategy honestly. Note that we are only considering deviations by a single party, as it is not possible for multiple parties to cheat in our model. In addition, the utility matrix satisfies honest invariance, in that the utility of the honest strategy profile remains unchanged for all parties. Now, in order to compute the deposits, we again apply \cref{lemma:Phi_must_be_left_invertible_to_implement_any_E} and compute the appropriate payment scheme:
$$
	\mathbf{\Lambda}_\texttt{PVC} = (\mathbf{U} - \mathbf{E})\,\mathbf{\Phi}_\texttt{PVC}^{-1} = \begin{pmatrix}
		0 & \frac{u^+ + \delta}{\varepsilon} & \delta & \frac{u^-}{\varepsilon} & 0 & \cdots & \frac{u^-}{\varepsilon} & 0 \\
		0 & \frac{u^-}{\varepsilon} & 0 & \frac{u^+ + \delta}{\varepsilon} & \delta  & \cdots & \frac{u^-}{\varepsilon} & 0 \\
		0 & \frac{u^-}{\varepsilon} & 0 & \frac{u^-}{\varepsilon} & 0  & \cdots & \frac{u^-}{\varepsilon} & 0 \\
		0 & \frac{u^-}{\varepsilon} & 0 & \frac{u^-}{\varepsilon} & 0  & \cdots & \frac{u^-}{\varepsilon} & 0 \\
		\vdots & \vdots & \vdots & \vdots & \vdots & & \vdots & \vdots \\
		0 & \frac{u^-}{\varepsilon} & 0 & \frac{u^-}{\varepsilon} & 0 & \cdots & 
		\frac{u^+ + \delta}{\varepsilon} & \delta
	\end{pmatrix}
$$

\noindent We now briefly analyze the resulting payment scheme. We note that when $\top$ is emitted, all parties are repaid their deposits in full. When the symbol $\textsf{abort}_i$ is emitted, the party $P_i$ loses part of their deposit, while all other parties are repaid their deposit in full. Finally, when $\textsf{cheat}_i$ is emitted, the party $P_i$ loses $\frac{u^+ + \delta}{\varepsilon}$, while each $P_j$ for $j \neq i$ loses $\frac{u^-}{\varepsilon}$. Note that we assume $u^- < 0$, meaning $P_i$ actually \emph{gains money} from the payment scheme, i.e. receive back more than they initially deposited. In order for the payment scheme to not mint new money, we need the following to hold true:
$$
	\frac{u^+ + \delta}{\varepsilon} \geq -\frac{(n-1)\,u^-}{\varepsilon}
$$
That is, we must have that $\delta \geq -(u^{+}+(n-1)\,u^-)\geq 0$ for the transformation to be implementable in practice, i.e. the payment scheme must be self-contained. In other words, there is only sufficient funds left over to compensate the honest parties, if the desired level of security is sufficiently high (and hence the deposits are large).

As $\varepsilon < 1$ and $\delta > 0$, we have $\frac{u^+ + \delta}{\varepsilon} > \delta$, which means we get a deposit of size $\lambda_i^* = \frac{u^+ + \delta}{\varepsilon}$. Note that the argument is fairly easy to adapt to the non-homogeneous setting, where we would instead get $\lambda_i^* = \frac{u^+_i + \delta}{\varepsilon}$, where $u^+_i$ is the utility gained by party $P_i$ when successful in cheating. This shows the following result.

\begin{theorem}[Rational MPC]
	Let $f$ be a public function and let $P_1, P_2, \ldots, P_n$ be a set of rational parties with the following utility function: namely, each $P_i$ earns $1$ utility by learning the output of the function, and $u_i^+$ utility from learning the inputs of the other parties, while they gain $u_i^-$ utility from another party learning their input. Then $f$ can be computed with $\delta$-strong game-theoretic security by augmenting any $\varepsilon$-deterrent PVC protocol with a payment scheme where $P_i$ makes a deposit of size $\geq \frac{u_i^+ + \delta - 1}{\varepsilon}$. The protocol is self-contained if only if for each $i$, $\delta \geq -\left(u_i^+ + \sum_{j \neq i} u_i^-\right)$.
\end{theorem}

In the next section, we show a general lower bound on the maximum deposit of any self-contained payment scheme, namely of size $\Omega(1+\delta \sqrt{n}/|\Sigma|)$. Note that this matches asymptotically the deposits in our MPC protocol, assuming the PVC protocol is fixed (and hence $\varepsilon,n,|\Sigma|$ are all constant).

\section{A Lower Bound on the Size of Payments}
\label{lower_bound}
In this section we prove a lower bound on the size of the largest payment necessary to achieve game-theoretic security. We show that the largest deposit must be linear in the security parameter $\varepsilon$, as well as linear in some of the utilities in the game.

To establish our bound, we use properties of matrix norms. We give a brief recap of matrix norms for the purpose of self-containment and refer to \cite{matrix_comp} for more details. We say a mapping $\norm{\cdot} : \mathbb{R}^{m \times n} \rightarrow \mathbb{R}$ is a \emph{matrix norm} if it satisfies the following properties for all matrices $\mathbf{A}, \mathbf{B} \in \mathbb{R}^{m \times n}$, and every scalar $\alpha \in \mathbb{R}$.
\begin{enumerate}
	\item (\emph{Positivity}). $\norm{\mathbf{A}} \geq 0$, and $\norm{\mathbf{A}} = 0$ iff $\mathbf{A} = \mathbf{0}$.
	\item (\emph{Homogeneity}). $\norm{\alpha \mathbf{A}} = |\alpha| \norm{\mathbf{A}}$.
	\item (\emph{Subadditivity}). $\norm{\mathbf{A} + \mathbf{B}} \leq \norm{\mathbf{A}} + \norm{\mathbf{B}}$.
\end{enumerate}
We denote by $\norm{\cdot}_p$ the matrix norm induced by the $L_p$ norm $\norm{\cdot}_p$ on vector spaces, and is defined as follows:
$$
	\norm{\mathbf{A}}_p = \sup_{\mathbf{x} \neq \mathbf{0}} \dfrac{\norm{\mathbf{A}\mathbf{x}}_p}{\norm{\mathbf{x}}_p}
$$
If in addition, $\norm{\mathbf{A}\mathbf{B}} \leq \norm{\mathbf{A}} \cdot \norm{\mathbf{B}}$, we say $\norm{\cdot}$ is \emph{submultiplicative}. It can be shown that $\norm{\cdot}_p$ is submultiplicative for any value of $p$. Some special cases that we will need are $p=1,2,\infty$ which can be characterized as follows. The quantity $\norm{\mathbf{A}}_1$ equals the maximum absolute column sum of the columns of $\mathbf{A}$, while the quantity $\norm{\mathbf{A}}_\infty$ gives the maximum absolute row sum of the rows of $\mathbf{A}$. Our lower bound is established by noting that we know these sums for the matrices used in our framework. An example of a matrix norm that is not submultiplicative is the max norm, $\norm{\mathbf{A}}_\text{max} = \max_{i,j} |\mathbf{A}_{ij}|$. However, we can relate this norm to $\norm{\cdot}_2$ using the following identity.
\begin{equation}\label{eq:max_L2}
	\norm{\mathbf{A}}_2 \geq \norm{\mathbf{A}}_\text{max} \geq \frac{\norm{\mathbf{A}}_2}{\sqrt{mn}}
\end{equation}
We will need the fact that all matrix norms are equivalent up to scalar multiple, in the sense that each pair of matrix norms $\norm{\cdot}_a, \norm{\cdot}_b$ are related by
$
\alpha \norm{\mathbf{A}}_a \leq \norm{\mathbf{A}}_b \leq \beta \norm{\mathbf{A}}_a,
$
for some constants $\alpha,\beta \in \mathbb{R}$. For our purposes, we need the following bounds:
\begin{align}
	\label{eq:L2_bounded_by_L1}
	\frac{1}{\sqrt{m}} \norm{\mathbf{A}}_1 &\leq \norm{\mathbf{A}}_2 \leq \sqrt{n} \norm{\mathbf{A}}_1\\
	\label{eq:L2_bounded_by_Linfty}
	\frac{1}{\sqrt{n}} \norm{\mathbf{A}}_\infty &\leq \norm{\mathbf{A}}_2 \leq \sqrt{m} \norm{\mathbf{A}}_\infty
\end{align}

\paragraph{Establish the lower bound.} Let $G$ be a fixed game with information structure $\langle \Sigma, \P \rangle$. Let $(\delta,t)$ be fixed, and let $\AA{t}, \e{t}$ be the corresponding constraints. We denote by $\alpha^{(t)}$ the number of rows in $\AA{t}$. Now, let $\L$ be any feasible payment scheme. We have already seen that any such $\L$ is a solution to the following equation:
\begin{equation}\label{eq:feasible_L}
	\AA{t} \L \P \leq \AA{t} \U - \e{t} 
\end{equation}
Applying \cref{eq:max_L2,eq:feasible_L} and the properties of $\norm{\cdot}_2$, we establish the following bound:
\begin{equation}\label{eq:Lmax_bound}
	\norm{\L}_\text{max} \geq \frac{1}{\sqrt{n\,|\Sigma|}} \left(\dfrac{\norm{\AA{t}\U}_2+\norm{\e{t}}_2}{\norm{\AA{t}}_2 \cdot \norm{\P}_2} \right)
\end{equation}
Each row of $\e{t}$ is filled with $\delta$, so the resulting absolute row sum is $\delta n$. Similarly, each row of $\AA{t}$ contains exactly one 1 and one -1, so each absolute row sum is 2. Finally, each column of $\P$ is a pdf, so its absolute row sum is 1. Combining these insights with \cref{eq:L2_bounded_by_L1,eq:L2_bounded_by_Linfty} and substituting in \cref{eq:Lmax_bound} gives the following bound:
\begin{align*}
\norm{\L}_\text{max} &\geq \frac{1}{\sqrt{n\,|\Sigma|}} \left(\dfrac{\norm{\AA{t}\U}_2+\sqrt{\alpha^{(t)}} \delta n}{2 \sqrt{\alpha^{(t)}} \cdot \sqrt{|\Sigma|}} \right)= \frac{1}{2\,|\Sigma|}\left(\delta\sqrt{n} + \frac{\norm{\AA{t}\U}_2}{\sqrt{n\,\alpha^{(t)}}}\right)\label{eq:Lmax_bound2}
\end{align*}
We note that in general, there is not much to say about $\norm{\AA{t}\U}_2$, as $\U$ can lie in the kernel of $\AA{t}$. This occurs if $\U$ already establishes exact $\delta$-strong $t$-robust game-theoretic security. 

Note that the bound, strictly speaking, is a bound on the largest \emph{absolute} deposit necessary to achieve security, while we are interested in bounding the largest \emph{positive} deposit, denoted instead by $\L_\text{max}^*$. If the game already is secure, the bound for the largest deposit should be zero, while the above bound is positive for any $\delta>0$. Indeed, $\norm{\L}_\text{max}\neq \L_\text{max}^*$ iff we can pay more to a party to misbehave \emph{and still retain security} than what we have to pay another party to behave properly. We note that this depends on the structure of the game and the intended strategy profile. In particular, it is independent of the security parameter. For this reason, we denote by $\Delta_G^{(t)}(s^*)$ the minmax deposit required to obtain 0-strong $t$-robust game-theoretic security. We note that $\Delta_G^{(t)} > 0$ iff the game is not secure for any $\delta \geq 0$, while $\Delta_G^{(t)} \leq 0$ iff the game is already secure for $\delta=0$.  

\noindent We note that by definition, $\Delta_G^{(t)}$ is a trivial lower bound on the size of the maximum deposit. We combine this with the above bounds to yield the following lower bound:

\begin{theorem}\label{thm:lower_bound}
	Let $G$ be a game on $n$ players with an information structure $\langle \Sigma, \P\rangle$, and let $s^*$ be the intended strategy profile. If $\L$ is self-contained and ensures $\delta$-strong $t$-robust game-theoretic security, then the maximum deposit must satisfy $	\L^*_\text{max} \geq \Delta_G^{(t)}(s^*) + \Omega\left(\frac{\delta\sqrt{n}}{|\Sigma|}\right)$.
\end{theorem}

\section{Conclusion and Future Work}
In this paper, we proposed a generic mechanism for incentivizing behavior in games by the use of payments. We analyzed the complexity of finding an optimal payment scheme and found it to be equivalent to linear programming. We demonstrated the applicability of our framework to concrete problems in distributed computing, namely decentralized commerce and secure multiparty computation. Finally, we proved a lower bound on the payments, showing that the largest payment must be linear in the security parameter for any self-contained payment scheme.

We hope that our framework will find applications in distributed computing, giving a simple, yet expressive model for instantiating payments in a variety of protocols. However, more attention is needed for games of imperfect information; we conjecture it is unlikely there is a generic and efficient way to find payments for such games, though a formal reduction would be ideal.

\bibliographystyle{alpha}
\bibliography{refs}

\end{document}